\documentclass{llncs}

\usepackage{geometry}
\usepackage{etex}
\usepackage{environ}
\usepackage{amsmath,amsfonts,amssymb}
\usepackage{xcolor}
\usepackage{graphicx}
\usepackage[binary-units]{siunitx}
\usepackage{numprint}
\usepackage[ruled,vlined]{algorithm2e}
\usepackage{pgfplots}\pgfplotsset{compat=1.14}
\usepackage{booktabs}
\usepackage{balance}

\usepackage[disable]{todonotes}

\newcommand{\LINCOS}{\ensuremath{\mathsf{LINCOS}}}
\newcommand{\ELSA}{\ensuremath{\mathsf{ELSA}}}

\newcommand{\pluseq}{\mathrel{{+}{=}}}

\newcommand{\andeq}{\mathrel{{\land}{=}}}

\newcommand{\NN}{\mathbb{N}}

\newcommand{\Commit}{\mathsf{Commit}}
\newcommand{\Verify}{{\sf Verify}}

\newcommand{\Stamp}{\mathsf{Stamp}}
\newcommand{\RenewTs}{\mathsf{RenewTs}}
\newcommand{\RenewCom}{\mathsf{RenewCom}}

\newcommand{\TS}{\mathsf{TS}}

\newcommand{\dat}{\mathsf{dat}}
\newcommand{\file}{\mathsf{file}}

\newcommand{\Com}{{\sf Com}}

\newcommand{\Sign}{{\sf Sign}}

\newcommand{\Extract}{{\sf Open}}
\newcommand{\Setup}{{\sf Setup}}
\newcommand{\tst}{\mathrm{ts}}

\newcommand{\Exp}{\mathsf{Exp}}

\newcommand{\Init}{\mathsf{Init}}
\newcommand{\Tree}{\mathsf{Tree}}

\newcommand{\Path}{\mathsf{Path}}

\newcommand{\name}{\mathsf{name}}
\newcommand{\names}{\mathtt{filenames}}
\newcommand{\verify}{\mathsf{verify}}
\newcommand{\store}{\mathsf{store}}
\newcommand{\comIds}{\mathtt{comIndices}}
\newcommand{\comCount}{\mathtt{comCount}}
\newcommand{\itemInfos}{\mathtt{itemInfos}}
\newcommand{\evidence}{\mathtt{evidence}}

\newcommand{\renewLists}{\mathtt{renewLists}}

\newcommand{\Forge}{\mathsf{Forge}}

\newcommand{\Store}{\mathsf{Store}}
\newcommand{\Retrieve}{\mathsf{Retrieve}}
\newcommand{\AddCom}{\mathsf{AddCom}}
\newcommand{\AddComRenew}{\mathsf{AddComRenew}}

\newcommand{\Clock}{\mathsf{Clock}}
\newcommand{\thetime}{\mathtt{time}}

\newcommand{\View}{\mathsf{View}}
\newcommand{\Share}{\mathsf{Share}}
\newcommand{\Reconstruct}{\mathsf{Reconstruct}}
\newcommand{\RenewShares}{\mathsf{RenewShares}}
\newcommand{\Reshare}{\mathsf{Reshare}}
\newcommand{\RenewSharing}{\mathsf{RenewSharing}}
\newcommand{\Break}{\mathsf{Break}}

\newcommand{\CD}{\mathsf{CD}}
\newcommand{\COM}{\mathsf{COM}}
\newcommand{\VCOM}{\mathsf{VC}}
\newcommand{\HT}{\mathsf{HT}}
\newcommand{\SHARE}{\mathsf{SHARE}}
\newcommand{\SIG}{\mathsf{SIG}}
\newcommand{\PKI}{\mathsf{PKI}}

\newcommand{\calA}{\mathcal{A}}
\newcommand{\calB}{\mathcal{B}}
\newcommand{\calM}{\mathcal{M}}

\newcommand{\calE}{\mathcal{E}}
\newcommand{\calS}{\mathcal{S}}
\newcommand{\calT}{\mathcal{T}}

\newcommand{\SH}{\mathsf{SH}}
\newcommand{\ES}{\mathsf{ES}}

\newtheorem{construction}{Construction}

\makeatletter

\makeatletter
\newcommand{\declaresaveclass}[2][1000]{%
 \expandafter\globtoksblk\csname ncory@class@#2\endcsname{#1}%
 \global\expandafter\toks\csname ncory@class@#2\endcsname{??}%
 \newcounter{ncory@class@count@#2}%
 \@namedef{thencory@class@count@#2}{%
   \number\csname ncory@class@#2\endcsname+\arabic{ncory@class@count@#2}\relax
 }%
}

\NewEnviron{savedenv}[2][]{%
\BODY
}

\newcommand\ncory@compute@v[1]{%
  \numexpr\csname ncory@class@#1\endcsname+\value{ncory@class@count@#1}\relax
}
\newcommand\ncory@compute@r[2]{%
  \numexpr\csname ncory@class@#2\endcsname+\getrefnumber{#1}\relax
}

\newcommand{\printsaved}[2]{%
  \toks0={??}%
  \the\toks\numexpr\getrefnumber{#1}\relax\par
 }

\newcommand{\printallsaved}[1]{%
 \@tempcnta=\z@
 \loop
  \ifnum\@tempcnta<\value{ncory@class@count@#1}
  \advance\@tempcnta\@ne
  \the\toks\numexpr\csname ncory@class@#1\endcsname+\@tempcnta\relax\par
  \repeat
 }
\makeatother

\declaresaveclass[100]{proof}
\declaresaveclass[100]{algorithm}

\begin{document}

\title{{\ELSA}: Efficient Long-Term Secure Storage of Large Datasets (Full Version)\thanks{This work has been co-funded by the \mbox{DFG} as part of project S6 within \mbox{CRC} 1119 \mbox{CROSSING}. A short version of this paper appears in the proceedings of \mbox{ICISC'18}.}}

\author{Matthias Geihs \and Johannes Buchmann}
\institute{TU Darmstadt, Germany}

\maketitle



\begin{abstract}
An increasing amount of information today is generated, exchanged, and stored digitally.
This also includes long-lived and highly sensitive information (e.g., electronic health records, governmental documents) whose integrity and confidentiality must be protected over decades or even centuries.
While there is a vast amount of cryptography-based data protection schemes, only few are designed for long-term protection.
Recently, Braun et al.\ (\mbox{AsiaCCS'17}) proposed the first long-term protection scheme that provides renewable integrity protection and information-theoretic confidentiality protection.
However, computation and storage costs of their scheme increase significantly with the number of stored data items.
As a result, their scheme appears suitable only for protecting databases with a small number of relatively large data items, but unsuitable for databases that hold a large number of relatively small data items (e.g., medical record databases).

In this work, we present a solution for \emph{efficient} long-term integrity and confidentiality protection of large datasets consisting of relatively small data items.
First, we construct a renewable vector commitment scheme that is information-theoretically hiding under selective decommitment.
We then combine this scheme with renewable timestamps and information-theoretically secure secret sharing.
The resulting solution requires only a single timestamp for protecting a dataset while the state of the art requires a number of timestamps linear in the number of data items.
We implemented our solution and measured its performance in a scenario where \numprint{12000} data items are aggregated, stored, protected, and verified over a time span of 100 years.
Our measurements show that our new solution completes this evaluation scenario an order of magnitude faster than the state of the art.
\end{abstract}



\section{Introduction}

\subsection{Motivation and problem statement}
Today, huge amounts of information are generated, exchanged, and stored digitally and these amounts will further grow in the future.
Much of this data contains sensitive information (e.g., electronic health records, governmental documents, enterprise documents) and requires protection of \emph{integrity} and \emph{confidentiality}.
Integrity protection means that illegitimate and accidental changes of data can be discovered. Confidentiality protection means that only authorized parties can access the data.
Depending on the use case, protection may be required for several decades or even centuries.
Databases that require protection are often complex and consist of a large number of relatively small data items that require continuous confidentiality protection and whose integrity must be verifiable independent from the other data items.

Today, integrity of digitally stored information is most commonly ensured using digital signatures (e.g., RSA \cite{Rivest:1978:MOD:359340.359342}) and confidentiality is ensured using encryption (e.g., AES \cite{AES}).
The commonly used schemes are secure under certain computational assumptions.
For example, they require that computing the prime factors of a large integer is infeasible.
However, as computing technology and cryptanalysis advances over time, computational assumptions made today are likely to break at some point in the future (e.g., RSA will become insecure once quantum computers are available \cite{shor1997polynomial}).
Consequently, \emph{computationally secure} cryptographic schemes
have a limited lifetime and are insufficient to provide \emph{long-term security}.

Several approaches have been developed to mitigate long-term security risks.
Bayer et al.\ \cite{Bayer1993} proposed a technique for prolonging the validity of a digital signature by using digital timestamps.
Based on their idea, a variety of long-term integrity protection schemes have been developed.
An overview of existing long-term integrity schemes is given by Vigil et al.\ in \cite{DBLP:journals/compsec/VigilBCWW15}.
%
In contrast to integrity protection, confidentiality protection cannot be prolonged.
There is no protection against an adversary that stores ciphertexts today, waits until the encryption is weakened, and then breaks the encryption and obtains the plaintexts.
Thus, if long-term confidentiality protection is required, then strong confidentiality protection must be applied from the start.
A very strong form of protection can be achieved by using \emph{information theoretically secure} schemes, which are invulnerable to computational attacks.
For example, key exchange can be realized using quantum key distribution \cite{RevModPhys.74.145}, encryption can be realized using one-time pad encryption \cite{BLTJ:BLTJ928},
and data storage can be realized using proactive secret sharing \cite{herzberg1995proactive}.
An overview of information theoretically secure solutions for long-term confidentiality protection is given by Braun et al.\ \cite{Braun2014}.

Recently, Braun et al. proposed {\LINCOS} \cite{Braun:2017:LSS:3052973.3053043}, which is the first long-term secure storage architecture that combines long-term integrity with long-term confidentiality protection.
While their system achieves high protection guarantees, it is only designed for storing and protecting a single large data object, but not databases that consist of a large number of small data items.
One approach to store and protect large databases with {\LINCOS} is to run an instance of {\LINCOS} for each data item in parallel.
However, with this construction the amount of work scales linearly with the number of stored data items.
Especially, if the database consists of a large number of relatively small data items, this introduces a large communication and computation overhead.

\subsection{Contribution}
In this paper we propose an efficient solution to storing and protecting large and complex datasets over long periods of time.
Concretely, we present the long-term secure storage architecture {\ELSA} that uses renewable vector commitments in combination with renewable timestamps and proactive secret sharing to achieve this.

Our first contribution (Section~\ref{sec.extvcom}) is to construct an extractable-binding and statistically hiding vector commitment scheme.
Such a scheme allows for committing to a large number of data items by a single short commitment.
The extractable binding property of the scheme enables renewable integrity protection \cite{buldas2017ltcom} while the statistical hiding property ensures information theoretic confidentiality.
Our construction is based on statistically hiding commitments and hash trees \cite{DBLP:conf/crypto/Merkle89}.
We prove that our construction is extractable binding given that the employed commitment scheme and hash function are extractable binding.
Furthermore, we prove that our construction is statistically hiding under selective opening, which guarantees that by opening the commitments to some of the data items no information about unopened data items is leaked.
The construction of extractable-binding and statistically hiding vector commitments may be of independent interest, for example, in the context of zero knowledge protocols \cite{gennaro2006independent}.

Our second contribution (Section~\ref{sec.elsa}) is the construction of the long-term secure storage architecture \ELSA{}, which uses our new vector commitment scheme construction to achieve efficient protection of large datasets.
While protecting a dataset with {\LINCOS} requires the generation of a commitment and a timestamp for each data item separately, \ELSA{} requires only a single vector commitment and a single timestamp to protect the same dataset.
Hence, the number of timestamps is decreased from linear in the number of data items to constant and this drastically reduces the communication and computation complexity of the solution.
Moreover, as the vector commitment scheme is hiding under selective decommitment, the integrity of stored data items can still be verified individually without revealing information about unopened data items.
{\ELSA} uses a separate service for storing commitments and timestamps, which allows for renewing the timestamp protection without access to the stored confidential data. The decommitments are stored together with the data items at a set of shareholders using proactive secret sharing.
We show that the long-term integrity security of $\ELSA$ can be reduced to the unforgeability security of the employed timestamp schemes and the binding security of the employed commitment schemes within their usage period. Long-term confidentiality security is based on the statistical hiding security of the employed commitment and secret sharing schemes.


Finally, we experimentally demonstrate (Section~\ref{sec.performance}) the performance improvements achieved by \ELSA{} in a scenario where \numprint{12000} data items of size \SI{10}{\kilo\byte} are aggregated, stored, protected, retrieved, and verified during a timespan of \numprint{100} years.
For this, we implemented \ELSA{} and the state of the art long-term secure storage architecture \LINCOS{}.
Our measurements show that {\ELSA} completes the evaluation scenario $17x$ faster than {\LINCOS} and integrity protection consumes $101x$ less memory.
In particular, protection renewal is significantly faster with \ELSA{}.
Renewing the timestamps for approximately \numprint{12000} data items takes \SI{21.89}{\minute} with \LINCOS{} and only \SI{0.34}{\second} with \ELSA{}.
Furthermore, storage of the timestamps and commitments consumes \SI{1.75}{\giga\byte} of storage space with \LINCOS{} and only \SI{17.27}{\mega\byte} with \ELSA{} at the end of the experiment.
These improvements are achieved at slightly higher storage costs for the shareholders.
Each shareholder consumes \SI{559}{\mega\byte} with \LINCOS{} and \SI{748}{\mega\byte} with \ELSA{}.
The storage costs for integrity protection are independent of the size of the data items.
Storage, retrieval, and verification of a data item takes less than a second.
Overall, our evaluation shows that \ELSA{} provides practical performance and is suitable for storing and protecting large and complex databases that consist of relatively small data items over long periods of time (e.g., health record or governmental document databases).

\subsection{Related work}
Our notion of vector commitments is reminiscent of the one proposed by Catalano and Fiore \cite{catalano2013vcom}. However, they do not consider the hiding property and therefore do not analyze hiding under selective opening security. Also, they do not consider extractable binding security.
Hofheinz \cite{Hofheinz2011} studied the notion of selective decommitment and showed that schemes can be constructed that are statistically hiding under selective decommitment.
However, they do not consider constructions of vector commitments where a short commitment is given for a set of messages.
In \cite{bitansky2017hunting}, Bitansky et al.\ propose the construction of a SNARK from extractable collision-resistant hash functions in combination with Merkle trees. While their construction is similar to the extractable-binding vector commitment scheme proposed in Section~\ref{sec.construction.extbind}, our construction relies on a weaker property (i.e., extractable-binding hash functions) and our security analysis provides concrete security estimates.

Weinert et al.\ \cite{Weinert:2017:MMP:3052973.3053025} recently proposed a long-term integrity protection scheme that also uses hash trees to reduce the number of timestamps. However, their scheme does not support confidentiality protection, lacks a formal security analysis, and is less efficient than our construction.
Only few work has been done with respect to combining long-term integrity with long-term confidentiality protection. The first storage architecture providing these two properties and most efficient to date is \LINCOS{} \cite{Braun:2017:LSS:3052973.3053043}.
Recently, another long-term secure storage architecture has been proposed by Geihs et al. \cite{geihs2018propyla} that provides access pattern hiding security in addition to integrity and confidentiality.
On a high level, this is achieved by combining \LINCOS{} with an information theoretically secure ORAM.
While access pattern hiding security is an interesting property in certain scenarios where meta information about the stored data is known, it is achieved at the cost of additional computation and communication and it is out of the scope of this work.



\section{Preliminaries}
\label{sec.prelims}

\subsection{Notation}
For a probabilistic algorithm $\calA$ and input $x$, we write $\calA(x) \to_r y$ to denote that $\calA$ on input $x$ produces $y$ using random coins $r$.
For a vector $V=(v_1,\ldots,v_n)$, $n \in \NN$, and set $I \subseteq [n]$, define $V_I:=(v_i)_{i \in I}$, and for $i \in [n]$, define $V_i:=v_i$.
For a pair of random variables $(A,B)$, we define the statistical distance of $A$ and $B$ as $\Delta(A,B):=\sum_x \left| \Pr_A(x) - \Pr_B(x) \right|$.

\subsection{Cryptographic primitives}
We describe the cryptographic primitives that will be relevant for the understanding of this paper.

\subsubsection{Digital Signature Schemes}
A digital signature scheme $\SIG$ is defined by a tuple $(\calM,\Setup,\Sign,\allowbreak \Verify)$, where $\calM$ is the message space, and $\Setup$, $\Sign$, $\Verify$ are algorithms with the following properties.

\begin{description}
\item[$\Setup \to (sk,pk)$:]
This algorithm generates a secret signing key $sk$ and a public verification key $pk$.

\item[$\Sign(sk,m) \to s$:]
This algorithm gets as input a secret key $sk$ and a message $m \in \calM$. It outputs a signature $s$.

\item[$\Verify(pk,m,s) \to b$:]
This algorithm gets as input a public key $pk$, a message $m$, and a signature $s$.
It outputs $b=1$, if the signature is valid, and $0$, if it is invalid.
\end{description}

\noindent
A signature scheme is $\epsilon$-secure \cite{Goldwasser:1988:DSS:45474.45480} if for any $t$-bounded algorithm $\calA$,
$$\Pr
\begin{bmatrix}
\Verify(pk,m,s)=1 \land m \not\in Q :\\
\Setup \to (sk,pk), \calA^O(pk) \to (m,s)
\end{bmatrix}
\leq \epsilon(t) \text,$$
where $O(m) =\{Q \pluseq m; \Sign(sk,m) \to s; \text{return }s;\}$.

\subsubsection{Timestamp schemes}
A timestamp scheme involves a client and a timestamp service.
The timestamp service initializes itself using algorithm $\Setup$.
The client uses protocol $\Stamp$ to request a timestamp from the timestamp service.
Furthermore, there exists an algorithm $\Verify$ that allows anybody to verify the validity of a message-timestamp-tuple.

Timestamp schemes can be realized in different ways (e.g., based on hash functions or digital signature schemes \cite{Haber1991}).
Here, we only consider timestamp schemes based on digital signatures.
This works as follows. On initialization, the timestamp service chooses a signature scheme $\SIG$ and runs the setup algorithm $\SIG.\Setup \to (sk,pk)$. It publishes the public key $pk$.
A client obtains a timestamp for a message $m$, as follows. First, it sends the message to the timestamp service. Then, the timestamp service reads the current time $t$ and creates a signature on $m$ and $t$ by running $\SIG.\Sign(sk,[m,t]) \to s$. It then sends the signature $s$ and the time $t$ back to the client.
Anybody can verify the validity of a timestamp $(t,s)$ for a message $m$ by checking $\SIG.\Verify(pk,[m,t],s)=1$.
Such a signature-based timestamp scheme is considered secure as long as the signature scheme is secure.

\subsubsection{Commitment schemes}
A (non-internactive) commitment scheme $\COM$ is defined by a tuple $(\calM,\Setup,\Commit,\Verify)$, where $\calM$ is the message space, and $\Setup$, $\Commit$, $\Verify$ are algorithms with the following properties.

\begin{description}
\item[$\Setup \to pk$:]
This algorithm generates a public commitment key $pk$.

\item[$\Commit(pk,m) \to (c,d)$:]
This algorithm gets as input a public key $pk$ and a message $m \in \calM$. It outputs a commitment $c$ and a decommitment $d$.

\item[$\Verify(pk,m,c,d) \to b$:]
This algorithm gets as input a public key $pk$, a message $m$, a commitment $c$, and a decommitment $d$.
It outputs $b=1$, if the decommitment is valid, and $0$, if it is invalid.
\end{description}

A commitment scheme is considered secure if it is hiding and binding.
There exist different flavors of defining binding security. Here, we are interested in extractable binding commitments as this enables renewable and long-term secure commitments \cite{buldas2017ltcom}.
A commitment scheme is $\epsilon$-extractable-binding-secure if for any $t_1$-bounded algorithm $\calA_1$, there exists an $t_\calE$-bounded algorithm $\calE$, such that for any $t_2$-bounded algorithm $\calA_2$,
$$\Pr
\begin{bmatrix}
\Verify(pk,m,c,d)=1 \land m \neq m^* :\\
\Setup \to pk,
\calA_1(pk) \to_r c, \\
\calE(pk,r) \to m^*,
\calA_2(k,r) \to (m,d)
\end{bmatrix}
\leq \epsilon(t_1,t_\calE,t_2) \text.$$

For any public key $k$ and message $m$, define $C_k(m)$ as the random variable that takes the value of $c$ when sampling $\Commit(k,m) \to (c,d)$.
A commitment scheme is $\epsilon$-statistically-hiding if for any $k \in \Setup$, any pair of messages $(m_1,m_2)$, $\Delta(C_k(m_1), C_k(m_2)) \leq \epsilon$.

\subsubsection{Keyed hash functions}
A keyed hash function is a tuple of algorithms $(K,H)$ with the following properties. $K$ is a probabilistic algorithm that generates a key $k$. $H$ is a deterministic algorithm that on input a key $k$ and a message $x \in \{0,1\}^*$ outputs a short fixed length hash $y \in \{0,1\}^l$, for some $l \in \NN$.
We say a keyed hash function $(K,F)$ is $\epsilon$-extractable-binding if for any $t_1$-bounded algorithm $\calA_1$, there exists a $t_\calE$-bounded algorithm $\calE$, such that for any $t_2$-bounded algorithm $\calA_2$,
$$\Pr_{K \to k}
\begin{bmatrix}
H(k,x)=H(k,x^*) \land x \neq x^* : \\
\calA_1(k) \to_r y, \calE(k,r) \to x^*,
\calA_2(k,r) \to x
\end{bmatrix}
\leq \epsilon(t_1,t_\calE,t_2) \text.$$

\subsubsection{Secret sharing schemes}
A secret sharing scheme allows a data owner to share a secret data object among a set of shareholders such that only specified subsets of the shareholders can reconstruct the secret, while all the other subsets of the shareholders have no information about the secret.
In this work, we consider threshold secret sharing schemes \cite{Shamir:1979:SS:359168.359176}, for which there exists a threshold parameter $t$ (chosen by the data owner) such that any set of $t$ shareholders can reconstruct the secret, but any set of less than $t$ shareholders has no information about the secret.
A secret sharing scheme has a protocol $\Setup$ for generating the sharing parameters, a protocol $\Share$ for sharing a data object, and a protocol $\Reconstruct$ for reconstructing a data object from a given set of shares.
In addition to standard secret sharing schemes, proactive secret sharing schemes additionally provide a protocol $\Reshare$ for protection against so called mobile adversaries \cite{herzberg1995proactive}.
The protocol $\Reshare$ is an interactive protocol between the shareholders after which all the stored shares are refreshed so that they no longer can be combined with the old shares for reconstruction. This protects against adversaries who gradually corrupt an increasing number of shareholders over the course of time.

\if01

\subsection{{\LINCOS}: Long-Term Secure Storage System providing Integrity and Confidentiality}\label{sec.lincos}

The long-term secure storage architecture {\LINCOS} proposed by Braun et al.\ \cite{Braun:2017:LSS:3052973.3053043} combines long-term integrity with long-term confidentiality protection.
%
It is organized in two subsystems, a confidentiality system for storing confidential data and an integrity system for protecting integrity (see Figure~\ref{fig.lincos}).

\begin{figure}
\centering
\includegraphics[width=.95\columnwidth]{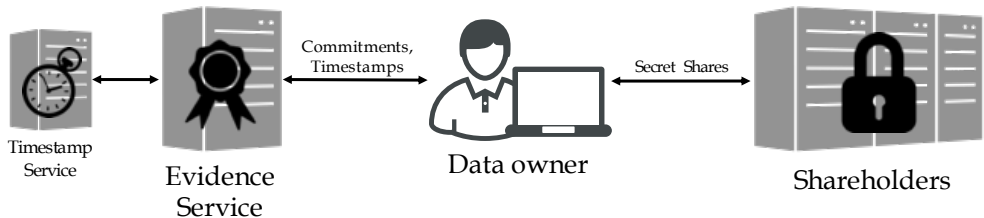}
\caption{\label{fig.lincos} Overview of the components of {\ELSA}.}
\end{figure}

\subsubsection{Confidentiality System.}
The purpose of the confidentiality system is to allow storage of confidential data at third party storage servers.
{\LINCOS} uses information theoretically secure proactive secret sharing \cite{DBLP:conf/crypto/HerzbergJKY95} to achieve information theoretic confidentiality protection.
The confidentiality system comprises of a set of storage servers that communicate with the data owner over information theoretically secure private channels.

\subsubsection{Integrity System.}
The integrity system of {\LINCOS} ensures long-term integrity protection of the stored data.
An integrity proof is generated and maintained in form of an attestation sequence, consisting of timestamps and commitments, that is updated regularly.


The integrity system of {\LINCOS} comprises of an evidence service and a timestamp service.
The evidence service maintains on behalf of the data owner the attestation sequence.
To preserve the confidentiality of the stored data, the data owner does not send the original data to the evidence service.
Instead, it generates information theoretically hiding commitments to the data and sends them.
The evidence service then periodically requests attestations for the commitments and the previous attestations from a timestamp service to ensure long-term protection.
The corresponding decommitment are stored by the data owner at the confidentiality system together with the protected data.
%

To ensure long-term integrity protection, timestamps and commitments have to be renewed periodically.
Timestamps are renewed by time-stamping the data together with the previous timestamps.
Commitments are renewed by committing to the data together with the previous decommitment values.
For more information we refer to the original work \cite{Braun:2017:LSS:3052973.3053043}.

\fi

\section{Statistically hiding and extractable binding vector commitments}
\label{sec.extvcom}

In this section, we define statistically hiding and extractable binding vector commitments, describe a construction, and prove the construction secure.
This construction is the basis for our performance improvements that we achieve with our new storage architecture presented in Section~\ref{sec.elsa}.

\subsection{Definition}

A vector commitment scheme allows to commit to a vector of messages $[m_1,\ldots,m_n]$.
It is extractable binding, if the message vector can be extracted from the commitment and the state of the committer and it is hiding under partial opening if an adversary cannot infer any valuable information about unopened messages, even if some of the committed messages have been opened.
Our vector commitments are reminiscent of the vector commitments introduced by Catalano and Fiore~\cite{catalano2013vcom}. However, neither do they require their commitments to be extractable binding nor do they consider their hiding property.\todo{their scheme additionally considers an update operation which we do not consider}

\begin{definition}[Vector commitment scheme]
A vector commitment scheme is a tuple $(L,\calM,\allowbreak\Setup,\Commit,\Extract,\Verify)$, where $L\in\NN$ is the maximum vector length, $\calM$ is the message space, and $\Setup$, $\Commit$, $\Extract$, and $\Verify$ are algorithms with the following properties.
\begin{description}
\item[$\Setup() \to k$:]
This algorithm generates a public key $k$.

\item[${\Commit(k,[m_1,\ldots,m_n]) \to (c,D)}$:]
On input key $k$ and message vector $[m_1,\ldots,m_n] \in \calM^n$, where $n \in [L]$, this algorithm generates a commitment $c$ and a vector decommitment $D$.

\item[$\Extract(k,D,i) \to d$:]
On input key $k$, vector decommitment $D$, and index $i$, this algorithm outputs a decommitment $d$ for the $i$-th message corresponding to $D$.

\item[$\Verify(k,m,c,d,i) \to b$:]
On input key $k$, message $m$, commitment $c$, decommitment $d$, and an index $i$, this algorithm outputs $b=1$, if $d$ is a valid decommitment from position $i$ of $c$ to $m$, and otherwise outputs $b=0$.
\end{description}
\end{definition}

A vector commitment scheme is correct, if a decommitment produced by $\Commit$ and $\Extract$ will always verify for the corresponding commitment and message.

\begin{definition}[Correctness]
A vector commitment scheme $(L,\calM,\allowbreak\Setup,\Commit,\Extract,\Verify)$ is correct if for all $n \in [L]$, $M \in \calM^n$, $k \in \Setup$, $i \in [n]$,
\begin{equation*}
\Pr\begin{bmatrix}
\Verify(k,M_i,c,d)=1:\\
\Commit(k,M) \to (c,D),
\Extract(k,D,i) \to d 
\end{bmatrix}=1
\text{ .}
\end{equation*}
\end{definition}

A vector commitment scheme is statistically hiding under selective opening, if the distribution of commitments and openings does not depend on the unopened messages.
\begin{definition}[Statistically hiding (under selective opening)]
Let $S=(L,\calM,\Setup,\allowbreak\Commit,\allowbreak\Extract,\Verify)$ be a vector commitment scheme.
For $n\in[L]$, $I \subseteq [n]$, $M\in\calM^n$, $k\in\Setup$, we denote by $\CD_k(M,I)$ the random variable $(c,\bar{D}_I)$, where $(c,D) \gets \Commit(k,M)$ and $\bar{D} \gets (\Extract(D,i))_{i\in [n]}$.
%
%
%
Let $\epsilon \in [0,1]$.
We say $S$ is $\epsilon$-statistically-hiding, if for all $n \in \NN$, $I \subseteq [n]$, $M_1,M_2\in\calM^n$ with $\left(M_1\right)_I=\left(M_2\right)_I$, $k \in \Setup$,
\begin{equation*}
\Delta(\CD_k(M_1,I),\CD_k(M_2,I)) \leq \epsilon
\text{ .}
\end{equation*}
\end{definition}

A vector commitment scheme is extractable binding, if for every efficient committer, there exists an efficient extractor, such that for any efficient decommitter, if the committer gives a commitment that can be opened by a decommitter, then the extractor can already extract the corresponding messages from the committer at the time of the commitment.
\begin{definition}[Extractable binding]
Let $\epsilon:\NN^3 \to [0,1]$.
We say a vector commitment scheme $(L,\calM,\Setup,\Commit,\Extract,\Verify)$ is $\epsilon$-extractable-binding, if for all $t_1$-bounded algorithms $\calA_1$, $t_\calE$-bounded algorithms $\calE$, and $t_2$-bounded algorithms $\calA_2$,
\begin{equation*}
\Pr\begin{bmatrix}
\Verify(p,m,c,d,i)=1
\land
m_i \neq m : \\
\Setup() \to k,
\calA_1(k) \to_r c,\\
\calE(k,r) \to [m_1,m_2,\ldots],
\calA_2(k,r) \to (m,c,d,i)
\end{bmatrix}
\leq
\epsilon(t_1,t_\calE,t_2)
\text{ .}
\end{equation*}
\end{definition}

\subsection{Construction: extractable binding}\label{sec.construction.extbind}
In the following, we show that the Merkle hash tree construction \cite{DBLP:conf/crypto/Merkle89} can be casted into a vector commitment scheme and that this construction is extractable binding if the used hash function is extractable binding.

\begin{construction}\label{const.ebvcom}
Let $(K,H)$ denote a keyed hash function and let $L \in \NN$.
The following is a description of the hash tree scheme by Merkle cast into the definition of vector commitments.

\begin{description}
\item[$\Setup() \to k$:]
Run $K \to k$ and output $k$.

\item[${\Commit(k,[m_1,\ldots,m_n]) \to (c,D)}$:]
Set $l \gets \min\{i \in \NN :n \leq 2^i\}$.
For $i \in \{0,\ldots,n-1\}$, compute $H(k,m_i) \to h_{i,l}$, and for $i \in \{n,\ldots,2^l-1\}$, set $h_{i,l} \gets \bot$.
For $i \in \{l-1,\ldots,0\}$, $j \in \{0,\ldots,2^i-1\}$, compute $H(k,[h_{i-1,2j} , h_{i-1,2j+1}])$.
Compute $H(k,[l,h_{0,0}]) \to c$.
Set $D \gets [h_{i,j}]_{i \in \{0,\ldots,l\}, j \in \{0,\ldots,2^i-1\}}$, and output $(c,D)$.

\item[$\Extract(k,D,i^*) \to d$:]
Let $D \to [h_{i,j}]_{i \in \{0,\ldots,l\}, j \in \{0,\ldots,2^i-1\}}$.
Set $a_l \gets i^*$.
For $j \in \{l,\ldots,1\}$, set $b_j \gets a_j + 2(a_j \bmod 2)-1$, $g_j \gets h_{j,b_j}$, and $a_{j-1} \gets \lfloor a_j/2\rfloor$.
Set $d=[g_1,\ldots,g_l]$ and output $d$.

\item[$\Verify(k,m,c,d,i^*) \to b^*$:]
Let $d=[g_1,\ldots,g_l]$.
Set $a_l \gets i^*$ and compute $H(k,m)\to h_l$.
For $i \in \{l,\ldots,1\}$, if $a_i \bmod 2=0$, set $b_i \gets [h_i,g_i]$, and if $a_i \bmod 2=1$, set $b_i \gets [g_i,h_i]$, and then compute $H(k,b_i) \to h_{i-1}$ and set $a_{i-1} \gets \lfloor a_i/2\rfloor$.
Compute $H(k,[l,h_0]) \to c'$.
Set $b^* \gets (c=c')$. Output $b^*$.
\end{description}
\end{construction}

\begin{theorem}\label{thm.ebvcom.correct}
The vector commitment scheme described in Construction~\ref{const.ebvcom} is correct.
\end{theorem}

\begin{savedenv}{proof}
\begin{proof}[Proof of Theorem~\ref{thm.ebvcom.correct}]
Let $(L,\calM,\Setup,\Commit,\Extract,\Verify)$ be the scheme described in Construction~\ref{const.ebvcom}.
Furthermore, let $n \in [L]$, $M \in \calM^n$, $k \in \Setup$, $i \in [n]$, $\Commit(k,M)\to(c,D)$, $\Extract(k,D,i)\to d$, and $\Verify(k,M_i,\allowbreak c,d) \to b$.
By the definition of algorithms $\Commit$ and $\Extract$, we observe that $d=[g_1,\ldots,g_l]$ are the siblings of the nodes in the hash tree $D=[h_{i,j}]_{i \in \{0,\ldots,l\}, j \in \{0,\ldots,2^i-1\}}$ on the path from leaf $H(k,M_i)$ to the root $h_{0,0}$.
We observe that algorithm $\Verify$ recomputes this path starting with $H(k,M_i)$.
Thus,  $h_0=h_{0,0}=c$. It follows that $b=1$.
\qed
\end{proof}
\end{savedenv}

\begin{theorem}\label{thm.ebvcom.eb}
Let $(K,H)$ be an $\epsilon$-extractable-binding hash function. The vector commitment scheme described in Construction~\ref{const.ebvcom} instantiated with $(K,H)$ is $\epsilon'$-extractable-binding with $\epsilon'(t_1,t_\calE,t_2)=2L*\epsilon(t_1 + t_\calE/L,t_\calE/L,t_2)$.
\end{theorem}

\begin{savedenv}{proof}
\begin{proof}[Proof of Theorem~\ref{thm.ebvcom.eb}]
Let $(K,H)$ be an $\epsilon$-extractable-binding keyed hash function and let $(L,\calM,\Setup,\Commit,\allowbreak\Extract,\Verify)$ be the vector commitment scheme described in Construction~\ref{const.ebvcom} instantiated with $(K,H)$.
To prove the theorem, we use the extractable-binding property of $(K,H)$ to construct an extractor for the vector commitment scheme.

Fix $t_1,t_\calE,t_2 \in \NN$ and $t_1$-bounded algorithm $\calA_1$.
%
We observe that, because $(K,H)$ is $\epsilon$-extractable-binding, for any $t$-bounded algorithm $\calA$, there exists a $t_\calE$-bounded algorithm $\calE^H_\calA$ such that for any $t_2$-bounded algorithm $\calA_2$, for experiment $\calA(k)\to_r h$, $\calA_2(k,r) \to m$, and $\calE^H_\calA(k,r) \to m'$, we have $H(k,m)=h$ and $m\neq m'$ with probability at most $\epsilon(t,t_\calE,t_2)$.

Define $\calA_{0,0}$ as an algorithm that on input $k$, samples $r$, computes $\calE^H_{\calA_1}(k,r) \to [l,h_{0,0}]$, and outputs $h_{0,0}$.
Recursively, define $\calA_{i,j}$ as an algorithm that on input $k$, samples $r$, computes $\calE^H_{\calA_{i-1,\lfloor{j/2}\rfloor}}(k,r) \to [h_0,h_1]$, and outputs $h_{j \bmod 2}$.

Now define the vector extraction algorithm $\calE$ as follows.
On input $(k,r)$, first compute $\calE^H_{\calA_1}(k,r)\to_{r'}[l,h_{0,0}]$.
Then, for $i \in \{0,\ldots,2^l-1\}$, sample $r_i$ and compute $\calE^H_{\calA_{l,i}}(k,[r,r',r_i])\to m_{i+1}$.
Output $[m_1,\ldots,m_{2^l}]$.

We observe that for any $t_2$-bounded algorithm $\calA_2$, for experiment $\calA_1(k)\to_r c$, $\calA_2(k,r) \to (m,d,i)$, and $\calE(k,r) \to M$, the probability of having $\Verify(k,m,c,d,i)=1$ and $M_i \neq m$ is upper-bounded by the probability that at least one of the node extraction algorithms relied on by $\calE$ fails.
As there are at most $2L$ nodes in the tree, this probability is upper-bounded by $2L*\epsilon(\max\{t_1, t_\calE\},\allowbreak t_\calE,\allowbreak t_2)$.

It follows that Construction~\ref{const.ebvcom} is $\epsilon'$-extractable-binding with $\epsilon'(t_1,t_\calE,t_2) = 2L*\epsilon(t_1 + t_\calE/L,t_\calE/L,t_2)$.
\end{proof}
\end{savedenv}

\subsection{Construction: extractable binding and statistically hiding}\label{sec.construction}
We now combine a statistically hiding and extractable binding commitment scheme with the vector commitment scheme from Construction~\ref{const.ebvcom} to obtain a statistically hiding (under selective opening) and extractable binding vector commitment scheme.
The idea is to first commit with the statistically hiding scheme to each message separately and then produce a vector commitment to these individually generated commitments.

\begin{construction}\label{const.ebshvcom}
Let $\COM$ be a commitment scheme and $\VCOM$ be a vector commitment scheme.
\begin{description}
\item[$\Setup() \to k$:]
Run $\COM.\Setup() \to k_1$, $\VCOM.\Setup() \to k_2$, set $k \gets (k_1,k_2)$, and output $k$.

\item[${\Commit(k,[m_1,\ldots,m_n]) \to (c,D)}$:]
Let $k \to (k_1,k_2)$.
For $i \in \{1,\ldots,n\}$, compute $\COM.\Commit(k_1,\allowbreak m_i) \to (c_i,d_i)$.
Then compute $\VCOM.\Commit(k_2,[c_1,\ldots,c_n]) \to (c,D')$, set $D \gets ([(c_1,d_1),\allowbreak\ldots,\allowbreak(c_n,d_n)],D')$, and output $(c,D)$.

\item[$\Extract(k,D,i) \to d$:]
Let $k \to (k_1,k_2)$ and $D \to ([(c_1,d_1),\allowbreak\ldots,\allowbreak(c_n,d_n)],D')$.
Compute $\VCOM.\Extract(k_2,\allowbreak D',i) \to d'$, set $d \gets (c_i,d_i,d')$, and output $d$.

\item[$\Verify(k,m,c,d,i) \to b$:]
Let $k \to (k_1,k_2)$ and $d \to (c',d',d'')$.
Compute $\COM.\Verify(k_1,\allowbreak m,c',d') \to b_1$ and then compute $\VCOM.\Verify(k_2,c',c,d'',i) \to b_2$, set $b \gets (b_1 \land b_2)$, and output $b$.
\end{description}
\end{construction}

\begin{theorem}\label{thm.ebshvcom.correct}
The vector commitment scheme described in Construction~\ref{const.ebshvcom} is correct if $\COM$ and $\VCOM$ are correct.
\end{theorem}

\begin{savedenv}{proof}
\begin{proof}[Proof of Theorem~\ref{thm.ebshvcom.correct}]
Let $(L,\calM,\Setup,\Commit,\Extract,\Verify)$ be the scheme described in Construction~\ref{const.ebshvcom}.
Let $n \in [L]$, $M \in \calM^n$, $k \in \Setup$, $i \in [n]$, $\Commit(k,M)\to(c,D)$, $\Extract(k,D,i)\to d$, and $\Verify(k,M_i,\allowbreak c,d) \to b$.
We observe that $D=((c_i,d_i)_{i\in[n]},D')$ and furthermore $d \in \VCOM.\Extract(k_2,\allowbreak D',i)$.
By the correctness of $\COM$, it follows that $\COM.\Verify(k_1,M_i,c_i,d_i)=1$, and by the correctness of $\VCOM$, it follows that $\HT.\VCOM(k_2,c_i,c,d,i)=1$.
Thus, $\Verify(k,M_i,c,d,i)=1$.
\end{proof}
\end{savedenv}

\begin{theorem}\label{thm.ebshvcom.hiding}
The vector commitment scheme described in Construction~\ref{const.ebshvcom} is $L\epsilon$-statistically-hiding (under selective opening) if the commitment scheme $\COM$ is $\epsilon$-statistically-hiding.
\end{theorem}

\begin{savedenv}{proof}
\begin{proof}[Proof of Theorem~\ref{thm.ebshvcom.hiding}]
Let $(L,\calM,\Setup,\Commit,\Extract,\Verify)$ be the scheme described in Construction~\ref{const.ebshvcom}.
For any $n \in [L]$, $I \subseteq [n]$, $M\in\calM^n$, $k\in\Setup$, and event $A \subseteq \Omega(\CD_k(M,I))$, denote by $\Pr_M(A)$ the probability that $A$ is observed when sampling from $\CD_k(M,I)$.

Let $n \in [L]$, $I \subseteq [n]$, $M_1,M_2\in\calM^n$ with $\left(M_1\right)_I=\left(M_2\right)_I$, $k \in \Setup$.
By the definition of the statistical distance, we have
\begin{equation*}
\Delta(\CD_k(M_1,I),\CD_k(M_2,I))
=
\sum_{c,\bar{D}_I} \left| \Pr_{M_1}(c,\bar{D}_I) - \Pr_{M_2}(c,\bar{D}_I)  \right|
\text{ .}
\end{equation*}
We observe that for any $M=(m_1,\ldots,m_n)\in\calM^n$, algorithm $\Commit(k,M)\to(c,D)$ first computes $\Com(k_1,m_i)\to(c_i,d_i)$, for each $i\in[n]$, and then computes $\HT.\Tree(k_2,[c_1,\ldots,c_n]) \to T$.
Denote $\tilde{C}=(c_i)_{i\in[n]}$ and $\tilde{D}=(d_i)_{i\in[n]}$.
By the law of total probability we have
\begin{equation*}
\Pr_{M}(c,\bar{D}_I)
=
\sum_{\stackrel{\tilde{C},\tilde{D}_I:}{c,\bar{D}_I}} \Pr_{M}(c,\bar{D}_I | \tilde{C},\tilde{D}_I) * \Pr_{M}(\tilde{C},\tilde{D}_I)
\text{ .}
\end{equation*}
Furthermore, we observe that $c$ and $T$ are determined by $(k,\tilde{C})$.
We also observe that on input $k$, $D=(T,(c_i,d_i)_{i\in[n]})$, and $i \in [n]$, algorithm $\Extract$ computes $\HT.\Path(k_2,T,i) \to P$ and outputs $d=(c_i,d_i,P)$. Hence, the output $d$ is determined by $(k,T,i,c_i,d_i)$.
Hence, $\Pr_{M}(c,\bar{D}_I | \tilde{C},\tilde{D}_I)>0$ implies $\Pr_{M}(c,\bar{D}_I | \tilde{C},\tilde{D}_I)=1$.
It follows that
\begin{equation*}
\sum_{\stackrel{\tilde{C},\tilde{D}_I:}{c,\bar{D}_I}} \Pr_{M}(c,\bar{D}_I | \tilde{C},\tilde{D}_I) * \Pr_{M}(\tilde{C},\tilde{D}_I)
=
\sum_{\stackrel{\tilde{C},\tilde{D}_I:}{c,\bar{D}_I}} \Pr_{M}(\tilde{C},\tilde{D}_I)
\text{ .}
\end{equation*}
Next, we observe from the description of algorithm $\Commit$ that each call $\Com(M_i)\to(c_i,d_i)$ is independent of the other calls $\Com(M_j)\to(c_j,d_j)$, $j \neq i$.
Thus,
\begin{equation*}
\Pr_{M}(\tilde{C}, \tilde{D}_I)
=
\Pr_{M}(\tilde{C}_I,\tilde{D}_I) * \prod_{i \in [n]\setminus I} \Pr_{M}(c_i)
\text{ .}
\end{equation*}
By the description of algorithms $\Commit$ and $\Extract$, we have that $(\tilde{C}_I,\tilde{D}_I)$ is determined by $M_I$.
We observe that $(M_1)_I=(M_2)_I$. Thus,
\begin{equation*}
\Pr_{M_1}(\tilde{C}_I,\tilde{D}_I)
=
\Pr_{M_2}(\tilde{C}_I,\tilde{D}_I)
\text{ .}
\end{equation*}
It follows that
\begin{equation*}
\begin{split}
&\left|
\Pr_{M_1}(\tilde{C}_I,\tilde{D}_I)
*
\left( \prod_{i \in [n]\setminus I} \Pr_{M_1}(c_i) \right)
-
\Pr_{M_2}(\tilde{C}_I,\tilde{D}_I)
*
\left( \prod_{i \in [n]\setminus I} \Pr_{M_2}(c_i) \right)
\right|
\\
&=
\Pr_{M_1}(\tilde{C}_I,\tilde{D}_I)
*
\left|
\left( \prod_{i \in [n]\setminus I} \Pr_{M_1}(c_i) \right)
-
\left( \prod_{i \in [n]\setminus I} \Pr_{M_2}(c_i) \right)
\right|
\\
&\leq
\Pr_{M_1}(\tilde{C}_I,\tilde{D}_I)
*
\sum_{i \in [n]\setminus I} \left| \Pr_{M_1}(c_i) - \Pr_{M_2}(c_i) \right|
\text{ .}
\end{split}
\end{equation*}
We observe that because $\Com$ is $\epsilon$-statistically-hiding, we have that
$\left| \Pr_{M_1}(c_i) - \Pr_{M_2}(c_i) \right| \leq \epsilon$.
It follows that
\begin{equation*}
\sum_{i \in [n]\setminus I} \left| \Pr_{M_1}(c_i) - \Pr_{M_2}(c_i) \right|
\leq
L\epsilon \text{ .}
\end{equation*}
In summary, we obtain
\begin{equation*}
\Delta(\CD_k(M_1,I),\CD_k(M_2,I))
\leq
L\epsilon
\text{ .}
\end{equation*}
\end{proof}
\end{savedenv}

\begin{theorem}\label{thm.ebshvcom.binding}
If $\COM$ and $\VCOM$ of Construction~\ref{const.ebshvcom} are $\epsilon$-extractable-binding, Construction~\ref{const.ebshvcom} is an $\epsilon'$-extractable-binding vector commitment scheme with $\epsilon'(t_1,t_\calE,t_2)=L*\epsilon(t_1 + t_\calE/L,t_\calE/L,t_2)$.
\end{theorem}

\begin{savedenv}{proof}
\begin{proof}[Proof of Theorem~\ref{thm.ebshvcom.binding}]
Let $t_1,t_\calE,t_2 \in \NN$ and fix any $t_1$-bounded algorithm $\calA_1$ that on input $k$ outputs $c$.

We observe that because $\VCOM$ is $\epsilon$-extractable-binding that there exists a $t_\calE$-bounded extraction algorithm $\calE_0$ such that for any $t_2$-bounded algorithm $\calA_2$, for experiment $\Setup \to k$, $\calA_1(k)\to_r c$, $\calA_2(k,r) \to (c_i,d,i)$, and $\calE_0(k,r) \to C$, we have $\VCOM.\Verify(k_2,\allowbreak c_i,c,d,i)=1$ and $C_i \neq c_i$ with probability at most $\epsilon(t_1,t_\calE,t_2)$.

Define $\calA_i$ as an algorithm that on input $k$, samples $r$, computes $\calE_0(k,r) \to C$, and outputs $C_i$.

For $i>0$, define $\calE_i$ as a $t_\calE$-bounded extraction algorithm such that for any $t_2$-bounded algorithm $\calA_2$, for experiment $K \to k$, $\calA_i(k) \to_r c_i$, $\calE_i(k,r) \to m_i'$, $\calA_2(k,r) \to (m_i,d_i)$, we have $\COM.\Verify(k_1,m_i,c_i,d_i)=1$ and $m_i \neq m_i'$ with probability at most $\epsilon(t_\calE,t_\calE,t_2)$.

Now we define a message vector extraction algorithm $\calE$ as follows.
On input $(k,r)$, first compute $\calE_1(k,r)\to_{r'}C$.
Then, for $i \in [|C|]$, compute $\calE_i(k,[r,r'])\to m_i$.
Output $[m_1,\ldots,m_{|C|}]$.

We observe that for any $t_2$-bounded algorithm $\calA_2$, for experiment $\calA_1(k)\to_r c$, $\calA_2(k,r) \to (m,d,i)$, and $\calE(k,r) \to M$, we have $\Verify(k,m,c,d,i)=1$ and $M_i \neq m$ with probability at most $L*\epsilon(t_1 + t_\calE,t_\calE,t_2)$.

It follows that the vector commitment scheme described in Construction~\ref{const.ebshvcom} is $\epsilon'$-extractable-binding with $\epsilon'(t_1,t_\calE,t_2) = L*\epsilon(t_1 + t_\calE/L,t_\calE/L,t_2)$.
\end{proof}
\end{savedenv}


\section{{\ELSA}: Efficient Long-Term Secure Storage Architecture}\label{sec.elsa}
Now we present {\ELSA}, a long-term secure storage architecture that efficiently protects large datasets.
It provides long-term integrity and long-term confidentiality protection of the stored data.
%
{\ELSA} uses statistically-hiding and extractable-binding vector commitments (as described in Section~\ref{sec.extvcom}) in combination with timestamps to achieve renewable and privacy preserving integrity protection.
The confidential data is stored using proactive secret sharing to guarantee confidentiality protection secure against computational attacks.
The data owner communicates with two subsystems (Figure~\ref{fig.elsa.overview}), where one is responsible for data storage with confidentiality protection and the other one is responsible for integrity protection.
The evidence service is responsible for integrity protection updates and the secret share holders are responsible for storing the data and maintaining confidentiality protection.
The evidence service also communicates with a timestamp service that is used in the process of evidence generation.

\begin{figure}\centering
\includegraphics[width=0.9\columnwidth,height=1.5in,keepaspectratio]{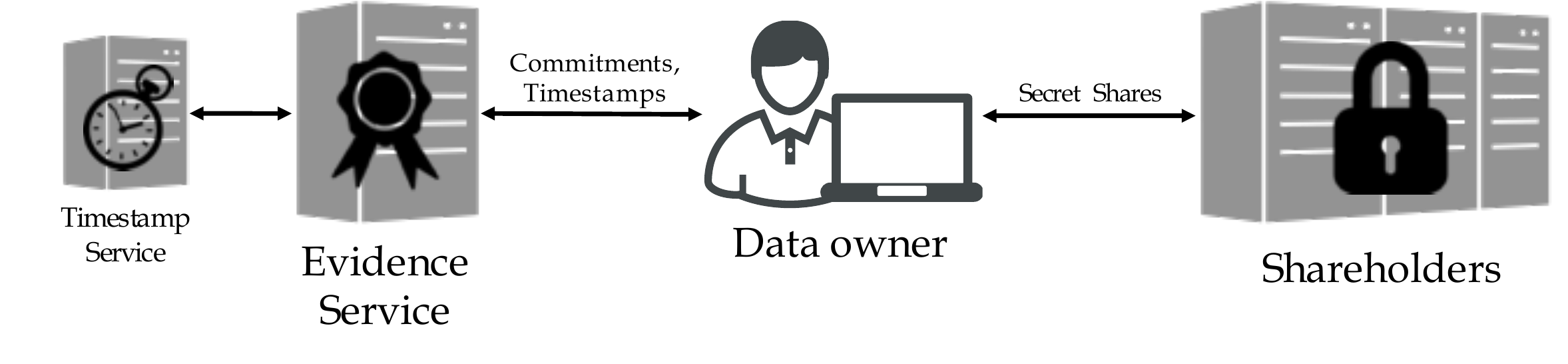}
\caption{Overview of the components of \ELSA{}.\label{fig.elsa.overview}}
\end{figure}

\subsection{Construction}
We now describe the storage architecture {\ELSA} in terms of the algorithms $\Init$, $\Store$, $\RenewTs$, $\RenewCom$, $\RenewShares$, and $\Verify$.
Algorithm $\Init$ initializes the architecture, $\Store$ allows to store new files, $\RenewTs$ renews the protection if the timestamp scheme security is weakened, $\RenewCom$ renews the protection if the commitment scheme security is weakened, $\RenewShares$ renews the shares to protect against a mobile adversary who collects multiple shares over time, and $\Verify$ verifies the integrity of a retrieved file.

We use the following notation. When we write $\SH.\Store(\name,\dat)$ we mean that the data owner shares the data $\dat$ among the shareholders using protocol $\SHARE.\Share$ associated with identifier $\name$.
If the shared data $\dat$ is larger then the size of the message space of the secret sharing scheme, $\dat$ is first split into chunks that fit into the message space and then the chunks are shared individually.
Each shareholder maintains a database that describes which shares belong to which data item name.
When we write $\SH.\Retrieve(\name)$, we mean that the data owner retrieves the shares associated identifier $\name$ from the shareholders and reconstructs the data using protocol $\SHARE.\Reconstruct$.

\subsubsection{Initialization}
The data owner uses algorithm $\ELSA.\Init$ (Algorithm~\ref{alg.init}) to initialize the storage system. The algorithm gets as input a proactive secret sharing scheme $\SHARE$, a set of shareholder addresses $(\mathrm{shURL}_i)_{i \in [N]}$, a sharing threshold $T$, and an evidence service address $\mathrm{esURL}$.
It then initializes the storage module $\SH$ by running protocol $\SHARE.\Setup$ and the evidence service module $\ES$ by setting $\ES.\evidence$ as an empty table and $\ES.\renewLists$ as an empty list.

\begin{savedenv}{algorithm}
\begin{algorithm}
\caption{\label{alg.init}$\ELSA.\Init(\SHARE,(\mathrm{shURL}_i)_{i \in [N]},T,\mathrm{esURL})$}
$\SH.\Init(\SHARE,(\mathrm{shURL}_i)_{i \in [N]},T)$\;
$\ES.\Init(\mathrm{esURL})$\;
\end{algorithm}
\end{savedenv}


\subsubsection{Data storage}
The client uses algorithm $\ELSA.\Store$ (Algorithm~\ref{alg.store}) to store a set of data files $[\file_i]_{i\in[n]}$, which works as follows.
First a signature scheme $\SIG$, a vector commitment scheme $\VCOM$, and a timestamp scheme $\TS$ are chosen.
Here, we assume that $\SIG$ is supplied with the secret key necessary for signature generation and $\VCOM$ is supplied with the public parameters necessary for commitment generation.
The algorithm first signs each of the data objects individually.
The algorithm then stores the file data, the public key certificate of the signature scheme instance, and the generated signature at the secret sharing storage system.
Afterwards, the algorithm generates a vector commitment $(c,D)$ to the file data vector and the signatures.
For each file, the corresponding decommitment is extracted and stored at the shareholders.
The file names $\names$, the commitment scheme instance $\VCOM$, the commitment $c$, and the chosen timestamp scheme instance $\TS$ are sent to the evidence service.
We remark that signing each data object individually potentially allows for having a different signer for each data object. If the data objects are to be signed by the same user, an alternative is to sign the commitment instead of the data objects.

When the evidence service receives $(\names,\VCOM,c,\TS)$, it does the following in algorithm $\AddCom$ (Algorithm~\ref{alg.es.add}).
It first timestamps the commitment $(\VCOM,c)$ and thereby obtains a timestamp $\tst$.
Then, it starts a new evidence list $l=[(\VCOM, c, \TS, \tst)]$ and assigns this list with all the file names in $\names$.
Also, it adds $l$ to the list $\renewLists$, which contains the lists that are updated on a timestamp renewal.

\begin{savedenv}{algorithm}
\begin{algorithm}
\caption{\label{alg.store}$\ELSA.\Store([\file_i]_{i\in[n]},\SIG,\VCOM,\TS)$}
$\names \gets \{\}$\;
\For{$i\in[n]$}
{
$\SIG.\Sign(\file_i.\dat) \to s_i$\;
$\SH.\Store([\text{'data'},\file_i.\name],[\file_i.\dat,\SIG.Cert,s_i])$\;
$\names \pluseq \file_i.\name$\;
}

$\VCOM.\Commit([\file_i.\dat,\SIG.Cert,s_i]_{i\in[n]}) \to (c,D)$\;

\For{$i\in[n]$}
{
$\VCOM.\Extract(D,i) \to d$\;
$\SH.\Store([\text{'decom'},\file_i.\name,i],d)$\;
}

$\ES.\AddCom(\names,\VCOM,c,\TS)$\;
\end{algorithm}
\end{savedenv}

\begin{savedenv}{algorithm}
\begin{algorithm}
\caption{\label{alg.es.add}$\ES.\AddCom(\names,\VCOM,c,\TS)$}
$\TS.\Stamp((\VCOM,c)) \to \tst$\;
$l \gets [(\VCOM, c, \TS, \tst)]$\;
\For{$\name \in \names$}
{
$\evidence[\name] \gets l$\;
$\renewLists \pluseq l$\;
}
\end{algorithm}
\end{savedenv}

\subsubsection{Timestamp renewal}
Algorithm $\ES.\RenewTs$ (Algorithm~\ref{alg.renewts}) is performed by the evidence service regularly in order to protect against the weakening of the currently used timestamp scheme.
The algorithm gets as input a vector commitment scheme instance $\VCOM'$ and a timestamp scheme instance $\TS$.
It first creates a vector commitment $(c',D')$ for the list of renewal items $\renewLists$.
Here, we only require the extractable-binding property of $\VCOM'$, while the hiding property is not required as all of the data stored at the evidence service is independent of the secret data due to the use of unconditionally hiding commitments by the data owner.
For each updated list item $i$, the freshly generated timestamp, commitment, and extracted decommitment are added to the corresponding evidence list $\renewLists[i]$.

\begin{savedenv}{algorithm}
\begin{algorithm}
\caption{\label{alg.renewts}$\ES.\RenewTs(\VCOM',\TS)$}

$\VCOM'.\Commit(\renewLists) \to (c',D')$\;
$\TS.\Stamp((\VCOM',c')) \to \tst$\;

\For{$i \in [|\renewLists|]$}
{
$\VCOM'.\Extract(D',i) \to d'$\;
$\renewLists[i] \pluseq (\VCOM',c',d',\TS,\tst)$\;
}
\end{algorithm}
\end{savedenv}

\subsubsection{Commitment renewal}
The data owner runs algorithm $\ELSA.\RenewCom$ (Algorithm~\ref{alg.renewcom}) to protect against a weakening of the currently used commitment scheme.
It chooses a new commitment scheme instance $\VCOM$ and a new timestamp scheme instance $\TS$ and proceeds as follows.
First the table of evidence lists $\ES.\evidence$ are retrieved from the evidence service and complemented with the decommitment values stored at the shareholders.
Next, a list with the data items, the signatures, and the current evidence for each data item is constructed.
This list is then committed using the vector commitment scheme $\VCOM$.
The decommitments are extracted and stored at the shareholders, and the commitment is added to the evidence at the evidence service using algorithm $\ES.\AddComRenew$.

\begin{savedenv}{algorithm}
\begin{algorithm}
\caption{\label{alg.renewcom}$\ELSA.\RenewCom(\VCOM,\TS)$}

$\comIds \gets \{\}$;
$\comCount \gets \{\}$;
$L \gets []$\;

\For{$\name \in \ES.\evidence$}
{
$\SH.\Retrieve([\text{'data'},\name]) \to (\dat,\SIG,s)$\;
$\ES.\evidence[\name] \to e$\;
\For{$i \in |e|$}
{
\If{$e_i.\VCOM \neq \bot$}
{$\SH.\Retrieve([\text{'decom'},\name,i]) \to e_i.d$\;}
}
$L \pluseq (\dat,\SIG,s,e)$\;
$\comIds[\name] \gets |L|$\;
$\comCount[\name] \gets |e|$\;
}

$\VCOM.\Commit(L) \to (c,D)$\;

\For{$\name \in \ES.\evidence$}
{
$\VCOM.\Extract(D,\comIds[\name]) \to d$\;
$\SH.\Store([\text{'decom'},\name,\comCount[\name]],d)$\;
}

$\ES.\AddComRenew(\VCOM,c,\TS)$\;
\end{algorithm}
\end{savedenv}

\begin{savedenv}{algorithm}
\begin{algorithm}
\caption{\label{alg.es.addcomrenew}$\ES.\AddComRenew(\VCOM,c,\TS)$}
$\TS.\Stamp((\VCOM,c)) \to \tst$\;
$l \gets [(\VCOM, c, \TS, \tst)]$\;
$\renewLists \gets [l]$\;
\For{$\name \in \evidence$}
{
$\evidence[\name] \pluseq l$\;
}
\end{algorithm}
\end{savedenv}


%
%
%
%


\subsubsection{Secret share renewal}

There are two types of share renewal supported by $\ELSA$.
The first type (Algorithm~\ref{alg.reshare}) triggers the share renewal protocol of the secret sharing system (i.e., the protocol $\SHARE.\Reshare$).
This interactive protocol refreshes the shares at the shareholders so that old shares, which may have leaked already, cannot be combined with the new shares, which are obtained after the protocol has finished, to reconstruct the stored data.
The second type (Algorithm~\ref{alg.resharing}) replaces the proactive sharing scheme entirely. This may be necessary if the scheme has additional security properties like verifiability (see proactive verifiable secret sharing \cite{herzberg1995proactive}), whose security may be weakened.
In this case, the data is retrieved, shared to the new shareholders, and finally the old shareholders are shutdown.

\begin{savedenv}{algorithm}
\begin{algorithm}
\caption{\label{alg.reshare}$\ELSA.\RenewShares()$}
$\SH.\Reshare()$\;
\end{algorithm}
\end{savedenv}

\begin{savedenv}{algorithm}
\begin{algorithm}
\caption{\label{alg.resharing}$\ELSA.\RenewSharing(\SHARE,(\mathrm{shURL}_i)_{i \in [N]},T)$}
$\SH'.\Init(\SHARE,(\mathrm{shURL}_i)_{i \in [N]},T)$\;
$I \gets \ES.\itemInfos$\;
\For{$\mathrm{name} \in I$}
{
$\SH.\Retrieve(\text{'data/'}+\mathrm{name}) \to \dat$\;
$\SH'.\Store(\text{'data/'}+\mathrm{name}, \dat)$\;
}
$\SH.\mathsf{Shutdown}()$\;
$\SH \gets \SH'$\;
\end{algorithm}
\end{savedenv}

\subsubsection{Data retrieval}
The algorithm $\ELSA.\Retrieve$ (Algorithm~\ref{alg.retrieve}) describes the data retrieval procedure of $\ELSA$.
It gets as input the name of the data file that is to be retrieved.
It then collects the evidence from the evidence service and the data from the shareholders.
Next, the evidence is complemented with the decommitments and then the algorithm outputs the data with the corresponding evidence.

\begin{savedenv}{algorithm}
\begin{algorithm}
\caption{\label{alg.retrieve}$\ELSA.\Retrieve(\name)$}
$e \gets \ES.\evidence[\name]$\;

\For{$i \in [|e|]$}
{
\If{$e_i.\VCOM \neq \bot$}{$\SH.\Retrieve([\text{'decom'},\name,i]) \to e_i.d$\;}
}

$\SH.\Retrieve([\text{'data'},\mathrm{name}]) \to (\dat,\SIG,s)$\;
$E \gets (\SIG,s,e)$\;

\KwRet $(\dat,E)$\;
\end{algorithm}
\end{savedenv}

\subsubsection{Verification}

Algorithm $\ELSA.\Verify$ (Algorithm~\ref{alg.verify}) describes how a verifier can check the integrity of a data item using the evidence produced by $\ELSA$.
Here we denote by $\mathrm{NTT}(i,e,t_\verify)$ the time of the next timestamp after entry $i$ of $e$ and by $\mathrm{NCT}(i,e,t_\verify)$ the time of the timestamp corresponding to the next commitment after entry $i$, and we set $\mathrm{NTT}(i,e,t_\verify) = t_\verify$ if $i$ is the last timestamp and $\mathrm{NCT}(i,e,t_\verify) = t_\verify$ if $i$ is the last commitment in $e$.
The algorithm gets as input a reference to the considered $\PKI$ (e.g., a trust anchor), the current verification time $t_\verify$, the data to be checked $\dat$, the storage time $t_\store$, and the corresponding evidence $E=(\SIG, s, e)$.
The algorithm returns true, if $\dat$ is authentic and has been stored at time $t_\store$.

In more detail, the verification algorithm works as follows.
It first checks whether the signature $s$ is valid for the data object $\dat$ under signature scheme instance $\SIG$ at the time of the first timestamp of the evidence list $e$.
It also checks whether the corresponding commitment is valid for $(\dat,\SIG,s)$ at the time of the next commitment and the timestamp is valid at the next timestamp.
Then, for each of the remaining $|e|-1$ entries of $e$, the algorithm checks whether the corresponding timestamp is valid at the time of the next timestamp and whether the corresponding commitments are valid at the time of the next commitments.
The algorithm outputs $1$ if all checks return valid, and it outputs $0$ in any other case.

\begin{savedenv}{algorithm}
\begin{algorithm}
\caption{\label{alg.verify}$\ELSA.\Verify(\PKI, t_\verify : \dat, t_\store, E) \to b$}
$(\SIG,s,e) \gets E$\;
$((\VCOM,c,d),(\VCOM',c',d'),(\TS,\tst)) \gets e_1$\;\vskip 1mm
$t_\mathrm{nt} \gets \mathrm{NTT}(1,e,t_\verify)$;
$t_\mathrm{nc} \gets \mathrm{NCT}(1,e,t_\verify)$\;
$b \gets \SIG.\Verify(\PKI, \tst.t : \dat, s)$\;
$b \andeq \VCOM.\Verify(\PKI, t_\mathrm{nc} : (\dat,\SIG,s), c, d)$\;
$b \andeq \TS.\Verify(\PKI, t_\mathrm{nt} : c, \tst, t_\store)$\;
$L \gets (\VCOM, c, \TS, \tst)$\;\vskip 1mm

\For{$i\in[2,\ldots,|e|]$}
{
$((\VCOM,c,d),(\VCOM',c',d'),(\TS,\tst)) \gets e_i$\;
$t_\mathrm{nt} \gets \mathrm{NTT}(i,e,t_\verify)$;
$t_\mathrm{nc} \gets \mathrm{NCT}(i,e,t_\verify)$\;

\uIf{$\VCOM=\bot$}{
$b \andeq \VCOM'.\Verify(\PKI, t_\mathrm{nt} : L, c', d')$\;
$b \andeq \TS.\Verify(\PKI, t_\mathrm{nt} : c', \tst, \tst.t)$\;
$L \pluseq (\VCOM',c',d',\TS,\tst)$\;
}
\Else{
$\dat' \gets (\dat,Cert,s,e[1,i-1])$\;
$b \andeq \VCOM.\Verify(\PKI, t_\mathrm{nc} : \dat', c, d)$\;
$b \andeq \TS.\Verify(\PKI, t_\mathrm{nt} : c, \tst, \tst.t)$\;
$L \gets (\VCOM, c, \TS, \tst)$\;
}

}

\KwRet $b$\;
\end{algorithm}
\end{savedenv}

\subsection{Security analysis}


\subsubsection{Computational model}
In a long-running system we have to consider adversaries that increase their computational power over time. For example, they may increase their computation speed or acquire new computational devices, such as quantum computers.
We capture this by using the computational model from \cite{buldas2017ltcom}.


A real-time bounded long-lived adversary $\calA$ is defined as a sequence $(\calA_{(0)}, \calA_{(1)},\allowbreak \calA_{(2)}, \ldots)$ of machines with $\calA_{(t)} \in \calM_t$ and is associated with a global clock $\Clock$.
We assume that the class $\calM_t$ of computing machines available at time $t$ widens when $t$ increases, i.e., $\calM_t \subseteq \calM_{t'}$ for $t<t'$.
The adversary is given the power to advance the clock, but it cannot go backwards in time.
When $\calA^\Clock$ is started, then actually the component $\calA_{(0)}$ is run.
Whenever a component $\calA_{(t)}$ calls the clock oracle to set a new time $t'$, then component $\calA_{(t)}$ is stopped and the component $\calA_{(t')}$ is run with input the internal state of $\calA_{(t)}$.
Real-time computational bounds are expressed as follows.
Let $\calA^\mathsf{Clock}$ be a computing machine associated with $\mathsf{Clock}$, and let $\rho:\NN \to \NN$ be a function.
We say $\calA$ is \emph{$\rho$-bounded} if for every time $t$, the aggregated step count of the machine components of $\calA$ until time $t$ is at most $\rho(t)$.

%

\subsubsection{Integrity}
First, we analyze integrity protection, by which we mean that it should be infeasible for an adversary to forge a valid evidence list for a data object, but the data object has not been authentically signed at the claimed time.
This is captured in the following definition, where we define security with respect to a set of schemes $\calS$ (e.g., commitment schemes, signature schemes, timestamp schemes) that are available within the context of the adversary. Also, we assume that each scheme instance $\calS_i$ is associated with a breakage time $t^b_i$ after which it is considered insecure.
The experiment (Algorithm~\ref{alg.elsa.unforge}) has a setup phase, where all the scheme instances are initialized by means of parameter generation. We write $\calS_i.\Setup() \to (sk,pk)$ to denote that a scheme potentially generates a secret parameter $sk$ (e.g., a private signing key) and a public parameter $pk$ (e.g., a public verification key or the parameters of a commitment scheme instance).
We allow the adversary to access the secret parameters of an instance once it is considered insecure (via oracle $\Break$).

\begin{definition}[Integrity]
We say {\ELSA} is $(\calM,\epsilon)$-unforgeable for schemes $\calS$, if for any $p$-bounded machine $\calA\in\calM$,
\begin{equation*}
\Pr\left[
\Exp^\Forge_\calS(\calA) = 1
\right]
\leq \epsilon(p) \text{ .}
\end{equation*}
The experiment $\Exp^\Forge$ is described in Algorithm~\ref{alg.elsa.unforge}.
\begin{algorithm}[h]
\caption{$\Exp^\Forge_\calS(\calA)$ \label{alg.elsa.unforge}}
\begin{minipage}{18em}
    $\mathsf{SetupExperiment}()$\;
    $\calA^{\Clock,\PKI,\SIG,\TS,\Break} \to (\dat,t_\store,E)$\;
    $t_\verify \gets \thetime$\;
    $b \gets \Verify(\PKI,t_\verify;\dat, t_\store, E) \land \dat \not\in Q_{t_\store}$\;
    \KwRet $b$\;
\end{minipage}
\hskip 1mm
\fbox{\begin{minipage}{13em}
\underline{$\mathsf{SetupExperiment}()$:}\\
$\thetime \gets 0$\;
\For{$i \in [|\calS|]$}{
$\calS_i.\Setup \to (sp,pp)$\;
$SP[i] \gets sp$;
$PP[i] \gets pp$\;
}
\end{minipage}}
\hskip 1mm
\fbox{\begin{minipage}{8em}
\underline{$\Clock(t)$:}\\
\If{$t>\thetime$}{$\thetime \gets t$\;}
\end{minipage}}
\hskip 1mm
\fbox{\begin{minipage}{7em}
\underline{$\PKI(i)$:}\\
\KwRet $PP[i]$\;
\end{minipage}}
\hskip 1mm
\fbox{\begin{minipage}{10em}
\underline{$\SIG(i,m)$:}\\
Assert $\calS_i.type=sig$\;
$Q[\thetime] \pluseq m$\;
$\calS_i.\Sign(SP[i];m) \to s$\;
\KwRet $s$\;
\end{minipage}}
\hskip 1mm
\fbox{\begin{minipage}{10em}
\underline{$\TS(i,m)$:}\\
Assert $\calS_i.type=ts$\;
$\Sign(i;[\thetime,m]) \to s$\;
\KwRet $(\thetime,s)$\;
\end{minipage}}
\hskip 1mm
\fbox{\begin{minipage}{9em}
\underline{$\Break(i)$:}\\
\If{$\thetime \geq \calS_i.tb$}{\KwRet $SP[i]$\;}
\end{minipage}}
\vskip 1mm
\end{algorithm}
\end{definition}

Next, we prove that the security of $\ELSA$ can be reduced to the extractable binding security of the used commitment schemes and the unforgeability of the used signature schemes within their validity period.

\begin{theorem}\label{thm.elsa.integrity}
Let $\calM=(\calM_t)_t$ specify which computational technology is available at which point in time and let $\calS$ be a set of cryptographic schemes, where each scheme $\calS_i \in \calS$ is associated with a breakage time $t^b_i$ and is $\epsilon_i$-secure against adversaries using computational technology $\calM_{t^b_i}$. In particular, we require unforgeability-security for signature schemes and extractable-binding-security for commitment schemes.
Let $p_E$ be any computational bound and $L$ be an upper bound on the maximum vector length of the commitment schemes in $\calS$.
Then, {\ELSA} is $(\calM,\epsilon)$-unforgeable for $\calS$ with $\epsilon(p)=(\sum_{i\in\SIG}\epsilon_i(p(t^b_i)p_E(t^b_i)L^2)) + (\sum_{i\in\COM}\epsilon_i(p(t^b_i),p_E(t^b_i),p(t^b_i)))$.
\end{theorem}

\begin{savedenv}{proof}
\begin{proof}[Proof of Theorem~\ref{thm.elsa.integrity}]
Suppose any $p$-bounded machine $\calA$ that interacts with interfaces $\Clock$, $\PKI$, $\SIG$, and $\TS$, and outputs $(\dat,\allowbreak t_\store,E)$.
For each signature scheme $i \in \SIG$, construct a machine $\calB_i$ with the goal to break the unforgeability of scheme $i$ until time $t^b_i$.
$\calB_i$ in the signature unforgeability experiment gets as input a public key $pk$ and access to a signing oracle $\Sign$. Its goal is to output $(m,s)$ such that $\calS_i.\Verify(pk,m,s)=1$ and the oracle $\Sign$ was not queried with $m$.
$\calB_i$, on input $pk$, does the following. It runs $\calA$ until time $t^b_i$ and simulates the environment of $\calA$ with the following difference. $\calB_i$ sets the public key of signature scheme $i$ to $pk$ and whenever the simulation of the experiment for $\calA$ requires the generation of a signature for scheme $i$, $\calB_i$ requests the signature from the oracle $\Sign$.
While $\calA$ is running, $\calB_i$ searches the outputs of $\calA$ for a valid message-signature-pair, where the message has not been queried to the signing oracle thus far.
$\calB_i$ also uses the extractable-binding property of the commitment schemes, as follows.
Whenever, $\calA$ queries the timestamp service $\TS$ with a commitment, then $\calB_i$ uses a $p_E$-bounded commitment message extractor to extract the corresponding messages out of $\calA$.
Let $L$ be an upper bound on the maximum length supported by all the used vector commitment schemes $\COM$.
Then, $\calB_i$ runs in at most $p_{\calB_i}=p(t^b_i)*p_E(t^b_i)*L^2$ operations and adheres to the computational model $\calM_{t^b_i}$.

We observe that an evidence object $E$ is a sequence $(\SIG,s,\allowbreak CDT_1,\allowbreak\ldots,\allowbreak CDT_n)$, where $CDT_i=(\VCOM_i,c_i,d_i,\VCOM'_i,c'_i,d'_i,\TS_i,\tst_i)$.
Define $CT_i=(\VCOM'_i,c'_i,d_i',\TS_i,\tst_i)$.
The verification algorithm of {\ELSA} ensures that $\tst_i$ is a valid timestamp for $(\VCOM_j,c_j,CT_j,\ldots,CT_{i-1})$, where $j$ is the largest index such that $\VCOM_j\neq\bot$ and $j \leq i$.
If $\VCOM_j\neq\bot$, then it also ensures that $d_i$ is a valid opening of $c_i$ to $(\dat,\SIG,s,CDT_1,\allowbreak\ldots,CDT_{i-1})$.
It follows that in every run in which $\calA$ outputs $(\dat,\allowbreak t_\store,E)$ with $\Verify(\PKI,t_\verify;\dat, t_\store,\allowbreak E)=1$ and $\dat \not\in Q_{t_\store}$, there is at least one $\calB_i$ that wins its unforgeability experiment before time $t^b_i$ or at least one of the extractors used by $\calB_i$ fails.
Hence, the probability that $\calA$ breaks {\ELSA} is upper bounded by $(\sum_{i\in\SIG}\epsilon_i(p_{\calB_i})) + (\sum_{i\in\COM}\epsilon_i(p(t^b_i),p_E(t^b_i),p(t^b_i)))$, where $\COM$ denotes the set of commitment schemes and $\SIG$ denotes the set of signature schemes of $\calS$.
\end{proof}
\end{savedenv}

\subsubsection{Confidentiality}
Next, we analyze confidentiality protection of $\ELSA$.
Intuitively, we require that an adversary with unbounded computational power who observes the data that is received by the evidence service and a subset of the shareholders does not learn any substantial information about the stored data. In particular, it should be guaranteed that an adversary does not learn anything about unopened files even if it retrieves some of the other files and the corresponding signatures, commitments, and timestamps.

We model this intuition by requiring that any (unbounded) adversary $\calA$ should not be able to distinguish whether it interacts with a system that stores a file vector $F_1$ or a system that stores another file vector $F_2$, if $\calA$ only opens a subset $I$ of files and $F_1$ and $F_2$ are identical on $I$.
This is modeled in experiment $\Exp^\mathsf{DIST}$ (Algorithm~\ref{alg.elsa.dist}), where we use the following notation.
For a secret sharing scheme $\SHARE$ we denote by $\SHARE.AS$ the set of authorized shareholder subsets that can reconstruct the secret.
For a protocol $\mathsf{P}$, we write $\langle \mathsf{P} \rangle_\View$ to denote an execution of $\mathsf{P}$ where $\View$ contains all the data sent and received by the involved parties. For an involved party $\mathcal{P}$, we write $\View(\mathcal{P})$ to denote the data sent and received by party $\mathcal{P}$.

\begin{definition}[Confidentiality]
We say {\ELSA} is $\epsilon$-statistically-hiding for $\calS$, if for all machines $\calA$, subsets $I$, sets of files $F_1,F_2$ with $(F_1)_I=(F_2)_I$, for all $L \in \NN$,
\begin{equation*}
\left|
\Pr\left[ \Exp^\mathsf{Dist}_{\calS,L}(\calA,F_1,I)=1 \right]
-
\Pr\left[ \Exp^\mathsf{Dist}_{\calS,L}(\calA,F_2,I)=1 \right]
\right| \leq \epsilon(L)
\text{ .}
\end{equation*}
The experiment $\Exp^\mathsf{DIST}$ is described in Algorithm~\ref{alg.elsa.dist}.

\begin{algorithm}[h]
\caption{$\Exp^\mathsf{Dist}_{\calS,L}(\calA,Files,I)$\label{alg.elsa.dist}}
\begin{minipage}{14em}
\begin{minipage}{9em}
    $\mathsf{SetupExperiment}()$\;
    $\calA^{\Clock,\ELSA'} \to b$\;
    \KwRet $b$\;
\end{minipage}
\vskip 1mm
\fbox{\begin{minipage}{13em}
\underline{$\mathsf{SetupExperiment}()$:}\\
$\thetime \gets 0$; $N \gets 0$\;
Generate parameters for $\calS$\;
$\ELSA.\Init(\calS.\SHARE,\SH,\ES)$\;
\end{minipage}}
\vskip 1mm
\fbox{\begin{minipage}{8em}
\underline{$\Clock(t)$:}\\
\If{$t>\thetime$}{$\thetime \gets t$\;}
\end{minipage}}
\end{minipage}
\hskip 1mm
\fbox{\begin{minipage}{25em}
\underline{$\ELSA'(op,param)$:}\\
\If{$N<L$}{
$N \pluseq 1$; $\calT \gets \calT'$\;
\uIf{$op=\Store \land param \not\in \calS.\SHARE.AS$}{
$\langle \ELSA.\Store(Files,param) \rangle_\View$;
$\calT \gets param$\;
}
\uElseIf{$op=\Retrieve \land param \in I$}{
$\langle \ELSA.\Retrieve(param) \rangle_\View$\;
}
\uElseIf{$op=ReTs$}{
$\langle \ELSA.\RenewTs(param) \rangle_\View$\;
}
\uElseIf{$op=ReCom$}{
$\langle \ELSA.\RenewCom(param) \rangle_\View$\;
}
\ElseIf{$op=ReShare \land param \not\in \calS.\SHARE.AS$}{
$\langle \ELSA.\RenewShares() \rangle_\View$;
$\calT' \gets param$\;
}
\KwRet $\View(\ES,\SH_\calT,Receiver)$\;
}
\end{minipage}}
\vskip 1mm
\end{algorithm}
\end{definition}

Next, we show that $\ELSA$ indeed provides confidentiality protection as defined above if statistically hiding secret sharing and commitment schemes are used.

\begin{theorem}\label{thm.elsa.confidentiality}
Let $\calS$ be a set of schemes, where $\calS.\SHARE$ is an $\epsilon$-statistically-hiding secret sharing scheme and every commitment scheme in $\calS$ is $\epsilon$-statistically-hiding.
Then, \ELSA{} is $\epsilon'$-statistically-hiding for $\calS$ with $\epsilon'(L)=2L\epsilon$.
\end{theorem}

\begin{savedenv}{algorithm}
\begin{proof}[Proof of Theorem~\ref{thm.elsa.confidentiality}]
We observe that the statistical distance between the view of the evidence service and the receiver for $F_1$ and the view for $F_2$ is at most $\epsilon$ per call to $\ELSA'$ because the vector commitment schemes in $\calS$ are $\epsilon$-statistically-hiding. The statistical distance between the view of an unauthorized subset of shareholders for $F_1$ and the view of the same subset of shareholders for $F_2$ is also at most $\epsilon$ because the secret sharing scheme is $\epsilon$-statistically-hiding.
It follows that the statistical distance for each query of the adversary diverges by at most $2\epsilon$.
Hence, overall the statistical distance is bounded by $2L\epsilon$.
\end{proof}
\end{savedenv}

With regards to protecting the network communication between the data owner and the shareholders we ideally require that information theoretically channels are used (e.g., based on Quantum Key Distribution and One-Time-Pad Encryption \cite{RevModPhys.74.145}), so that a network provider, who could potentially intercept all of the secret share packages, cannot learn any information about the communicated data. If information theoretically secure channels are not available, we recommend to use very strong symmetric encryption (e.g., AES-256 \cite{AES}).


\section{Performance evaluation}
\label{sec.performance}
We compare the performance of our new architecture \ELSA{} with the performance of the storage architecture {\LINCOS} \cite{Braun:2017:LSS:3052973.3053043}, which is the fastest existing storage architecture that provides long-term integrity and long-term confidentiality.

\subsection{Evaluation scenario}
For our evaluation we consider a scenario inspired by the task of securely storing electronic health records in a medium sized doctor's office.
The storage time frame is \numprint{100} years.
Every month, 10 new data items of size \SI{10}{\kilo\byte} (e.g., prescription data of patients) are added.
Every year, one document from each of the previous years is retrieved and verified (e.g., historic prescription data is read from the archives).

We assume the following renewal schedule for protecting the evidence against the weakening of cryptographic primitives.
The signatures are renewed every 2 years, as this is a typical lifetime of a public key certificate, which is needed to verify the signatures.
While signature scheme instances can only be secure as long as the corresponding private signing key is not leaked to an adversary, commitment scheme instances do not involve the usage of any secret parameters.
Therefore, their security is not threatened by key leakage and we assume that they only need to be renewed every 10 years in order to adjust the cryptographic parameter sizes or to choose a new and more secure scheme.

In our architecture we instantiate signature and commitment schemes as follows.
As signature scheme, we first use the RSA Signature Scheme~\cite{Rivest:1978:MOD:359340.359342} and then switch to the post-quantum secure XMSS signature scheme \cite{Buchmann2011} by 2030, as we anticipate the development of large-scale quantum computers.
Both of these schemes satisfy unforgeability security as required by Theorem~\ref{thm.elsa.integrity}.
As the vector commitment scheme we use Construction~\ref{const.ebshvcom} with the statistically hiding commitment scheme by Halevi and Micali~\cite{Halevi1996} whose security is based on the security of the used hash function which we instantiate with members of the SHA-2 hash function family~\cite{sha2}.
If we model hash functions as random oracles, the extractable-binding property required by Theorem~\ref{thm.elsa.integrity} is provided.
This vector commitment scheme construction also provides statistical hiding security as required by Theorem~\ref{thm.elsa.confidentiality}.
%
%
We adjust the signature and commitment scheme parameters over time as proposed by Lenstra and Verheul~\cite{DBLP:journals/joc/LenstraV01,lenstra2004key}. The resulting parameter sets are shown in Table~\ref{tab.params}.

\begin{table}
\center

\caption{\label{tab.params} Overview of the used commitment and signature scheme instances and their usage period.}

\begin{tabular}{ccc}
\toprule
\textbf{Years} & \textbf{Signatures} & \textbf{Commitments} \\
\midrule

2018-2030 & RSA-2048 & HM-256 \\

%

2031-2090 & XMSS-256 & HM-256 \\

2091-2118 & XMSS-512 & HM-512 \\

\bottomrule
\end{tabular}

\end{table}

For the storage system, we use the secret sharing scheme by Shamir \cite{Shamir:1979:SS:359168.359176}.
We run this scheme with $4$ shareholders and a threshold of $3$ shareholders are required for reconstruction.
Secret shares are renewed every 5 years, where the resharing is carried out centrally by the data owner.

\subsection{Results}
\todo{currently we do not account for the initial signature (with respect to signature generation time, storage space, and verification time)}%
We now present the results of our performance analysis.
Figure~\ref{fig.eval.overalltime} compares the computation time and storage costs of the two systems, {\ELSA} and {\LINCOS}.
Our implementation was done using the programming language Java.
The experiments were performed on a computer with a quad-core AMD Opteron CPU running at \SI{2.3}{\giga\hertz} and the Java virtual machine was assigned \SI{32}{\giga\byte} of RAM.

We observe that $\ELSA$ is much more computationally efficient compared to $\LINCOS$.
Completing the experiment using $\LINCOS$ took approximately \SI{6.81}{\hour}, while it took only \SI{24}{\minute} using $\ELSA$.
The biggest difference in the timings is observed when renewing timestamps.
Timestamp renewal with \LINCOS{} for year 2116 takes \SI{21.89}{\minute}, while it takes only \SI{0.34}{\second} with $\ELSA$.
Data storage performance is also considerably faster with $\ELSA$ than with $\LINCOS$.
The same holds for the commitment renewal procedure.
Data retrieval and verification performance is similar for the two systems.

Next, we observe that $\ELSA$ is also more efficient compared to $\LINCOS$ when it comes to the consumed storage space at the evidence service.
This is, again, because \ELSA{} requires fewer timestamps to be generated and stored than \LINCOS{}.
After running for 100 years, the evidence service of $\ELSA$ consumes only \SI{17.27}{\mega\byte} while the evidence service of $\LINCOS$ consumes \SI{1.75}{\giga\byte} of storage space.
We observe by Figure~\ref{fig.eval.space.sh} that $\ELSA$ consumes slightly more storage space at the shareholders than $\LINCOS$. This is because additional decommitment information for the vector commitments must be stored.
After running for 100 years, a shareholder of $\ELSA$ consumes about \SI{748}{\mega\byte} while a shareholder of $\LINCOS$ consumes about \SI{559}{\mega\byte} of storage space.

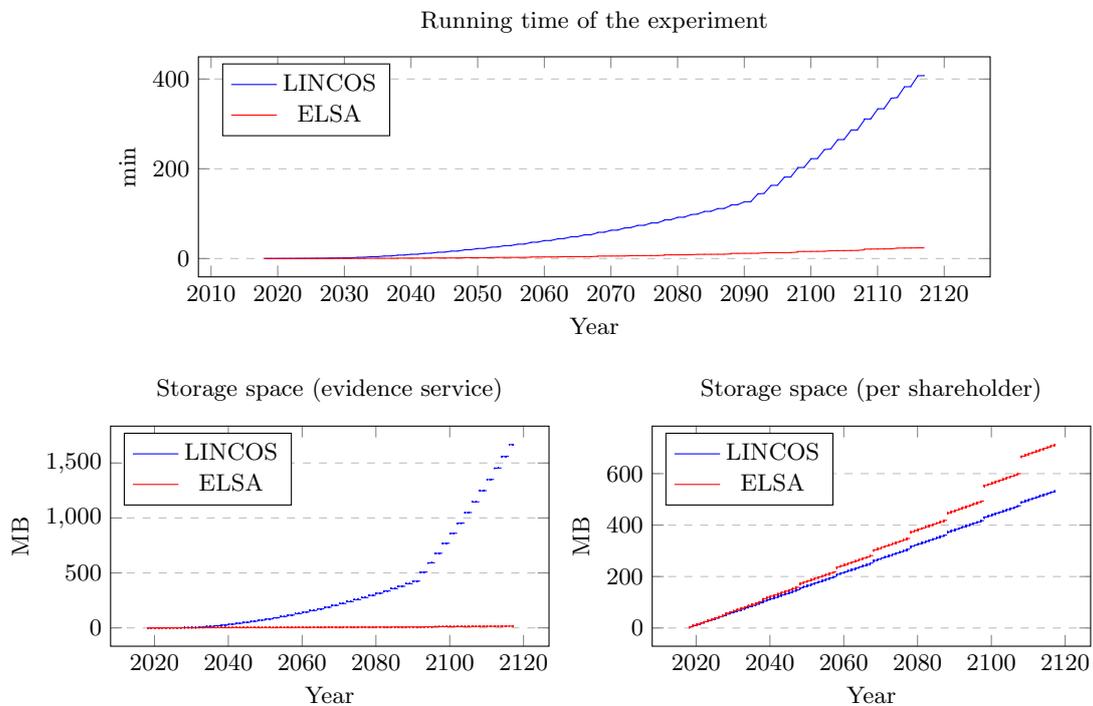
\begin{figure}
\centering

\begin{tikzpicture}
\begin{axis}[
    title={Running time of the experiment},
    xlabel={Year},
    ylabel={\si{\minute}},
    ymajorgrids=true,
    grid style=dashed,
    width=0.8\columnwidth,
    height=4.5cm,
    x tick label style={/pgf/number format/.cd, set thousands separator={}},
    legend pos=north west,
    y filter/.code={\pgfmathparse{#1/1000/60}\pgfmathresult},
]

\addplot+[no marks]
table[skip first n=16, x index=0, y index=6]
{data/results_LINCOS.txt};
\addlegendentry{LINCOS}

\addplot+[no marks]
table[skip first n=16, x index=0, y index=6]
{data/results_ELSA.txt};
\addlegendentry{ELSA}

\end{axis}
\end{tikzpicture}

\vskip1em


%
%
%
%
%
%


\begin{tikzpicture}
\begin{axis}[
    title={Storage space (evidence service)},
    xlabel={Year},
    ylabel={\si{\mega\byte}},
    ymajorgrids=true,
    grid style=dashed,
    width=0.49\columnwidth,
    height=4.5cm,
    x tick label style={/pgf/number format/.cd, set thousands separator={}},
    legend pos=north west,
    y filter/.code={\pgfmathparse{#1/1024/1024}\pgfmathresult},
]

\addplot+[no marks, jump mark left]
table[skip first n=16, x index=0, y index=4]
{data/results_LINCOS.txt};
\addlegendentry{LINCOS}

\addplot+[no marks, jump mark left]
table[skip first n=16, x index=0, y index=4]
{data/results_ELSA.txt};
\addlegendentry{ELSA}

\end{axis}
\end{tikzpicture}
%
%
%
\begin{tikzpicture}
\begin{axis}[
    title={Storage space (per shareholder)},
    xlabel={Year},
    ylabel={\si{\mega\byte}},
    ymajorgrids=true,
    grid style=dashed,
    width=0.49\columnwidth,
    height=4.5cm,
    x tick label style={/pgf/number format/.cd, set thousands separator={}},
    legend pos=north west,
    y filter/.code={\pgfmathparse{#1/1024/1024}\pgfmathresult},
]

\addplot+[no marks, jump mark left]
table[skip first n=16, x index=0, y index=5]
{data/results_LINCOS.txt};
\addlegendentry{LINCOS}

\addplot+[no marks, jump mark left]
table[skip first n=16, x index=0, y index=5]
{data/results_ELSA.txt};
\addlegendentry{ELSA}

\end{axis}
\end{tikzpicture}


\caption{Running time of the experiment and storage space consumption of the evidence service and per shareholder.}
\label{fig.eval.overalltime}\label{fig.eval.space.es}\label{fig.eval.space.sh}
\end{figure}

\bibliographystyle{bib/splncs04}
\balance
\bibliography{bib/bibliography}


\end{document}